\documentclass[conference]{IEEEtran}

\usepackage{cite}
\usepackage{graphicx}
\usepackage{amsmath}
\usepackage{latexsym}
\usepackage{subfigure}
\usepackage{balance}

\usepackage{eso-pic}

\AddToShipoutPicture*{\small
\sffamily\raisebox{1.8cm}{\hspace{1.8cm}
978-1-4244-2677-5/08/\$25.00
2008 IEEE }}

\begin{document}

\newtheorem{theorem}{Theorem}[section]
\newtheorem{lemma}[theorem]{Lemma}
\newtheorem{conjecture}[theorem]{Conjecture}
\newtheorem{corollary}[theorem]{Corollary}
\newtheorem{definition}[theorem]{Definition}
\newtheorem{scheme}[theorem]{Scheme}

\title{On the Capacity Improvement of Multicast Traffic with Network Coding}

\author{\IEEEauthorblockN{Zheng Wang\IEEEauthorrefmark{2},
Shirish Karande\IEEEauthorrefmark{2}, Hamid R.
Sadjadpour\IEEEauthorrefmark{2}, J.J.
Garcia-Luna-Aceves\IEEEauthorrefmark{3}}
\IEEEauthorblockA{Department of Electrical
Engineering\IEEEauthorrefmark{2} and Computer
Engineering\IEEEauthorrefmark{3}
\\
University of California, Santa Cruz, 1156 High Street, Santa Cruz, CA 95064, USA\\
$^\ddagger$ Palo Alto Research Center (PARC), 3333 Coyote Hill Road,
Palo Alto, CA 94304, USA \\
 Email:{\{wzgold, karandes,
hamid, jj\}@soe.ucsc.edu}} }

\maketitle

\begin{abstract}

In this paper, we study the contribution of network coding (NC) in
improving the multicast capacity of random wireless ad hoc networks
when nodes are endowed with multi-packet transmission (MPT) and
multi-packet reception (MPR) capabilities. We show that a per session
throughput capacity of $\Theta\left(nT^{3}(n)\right)$, where $n$ is
the total number of nodes and $T(n)$ is the communication range, can be
achieved as a tight bound when each session contains a constant number
of sinks. Surprisingly, an identical order capacity can be achieved
when nodes have only MPR and MPT capabilities. This result proves that
NC does not contribute to the order capacity of multicast traffic in
wireless ad hoc networks when MPR and MPT are used in the network. The
result is in sharp contrast to the general belief (conjecture) that NC
improves the order capacity of multicast. Furthermore, if the
communication range is selected to guarantee the connectivity in the
network, i.e., $T(n)\ge \Theta\left(\sqrt{\log n/n}\right)$, then the
combination of MPR and MPT achieves a throughput capacity of
$\Theta\left(\frac{\log^{\frac{3}{2}} n}{\sqrt{n}}\right)$ which
provides an order capacity gain of $\Theta\left(\log^2 n\right)$
compared to the point-to-point multicast capacity with the same number
of destinations.

\end{abstract}
\IEEEpeerreviewmaketitle

\section{Introduction}

The seminal work by Gupta and Kumar \cite{GuKu00} has sparked a
growing amount of interest in understanding the fundamental capacity
limits of wireless ad hoc networks. Several techniques \cite{OLT07,
MSG07, GMPS07} have been developed with the objective of improving the
capacity of wireless ad hoc networks. Network coding (NC), which was
originally proposed by Ahlswede et al. in \cite{ANLY00}, is one such
technique. Unlike traditional store-and-forward routing, network
coding encodes the messages received at intermediate nodes,
prior to forwarding them to subsequent next-hop neighbors. Ahlswede et
al. \cite{ANLY00} showed that network coding can achieve a multicast
flow equal to the min-cut for a single source and under the
assumptions of a directed graph.  This and other work in
NC \cite{LiYeCa03, KoMe03} has motivated a large number of
researchers to investigate the impact of NC in increasing the
throughput capacity of wireless ad hoc networks. However, Liu et
al. \cite{LGT07} recently showed that NC does not increase the order
of the throughput capacity for multi-pair unicast
traffic. Nevertheless, a number of efforts (analog network
coding~\cite{KGK07}, physical network coding~\cite{ZLL06}) have
continued the quest for improving the multicast capacity of ad-hoc
networks by using NC.  Despite the claims of throughput improvement by
such studies, the impact of NC on the multicast scaling law has
remained uncharacterized.

Promising approaches \cite{KGK07, ZLL06} implicitly assume the
combination of NC  with Multi-packet Transmission (MPT) and
Multi-packet Reception (MPR) \cite{RSW05,AKMK07,KASYK07} (i.e., the
ability to transceive successfully multiple concurrent transmissions
by employing physical-layer interference cancelation techniques).
MPR has been shown to increase the capacity regions of ad hoc
networks~\cite{ToGo03}, and very recently Garcia-Luna-Aceves et
al.~\cite{GSW07b} have shown that the order capacity in wireless ad
hoc networks subject to multi-pair unicast traffic is increased with
MPR. These prior efforts raise three important following questions:
(a) What is the multicast throughput order achieved by the
combination of NC with MPT and MPR? (b) Does this combination
provide us with an order gain over traditional techniques based on
routing and point-to-point communication? (c) If yes, what exactly
leads to this gain? Is NC necessary or does the combination of MPT
and MPR suffice?

In this work, we address the above three questions. The answers can be
summarized by our main results:
\begin{itemize}
\item When each multicast group consists of a constant number of
  sinks, the combination of NC, MPT and MPR provides a per session
  throughput capacity of $\Theta(nT^3(n))$, where $T(n)$ is the
  communication range.
\item This scaling law represents an order gain of $\Theta(n^2T^4(n))$ over a combination of routing and single packet transmission/reception.
\item The combination of only MPT and MPR is sufficient to achieve
  a per-session multicast throughput order of $\Theta(nT^3(n))$.
  Consequently, NC does not contribute to the multicast capacity when MPR and MPT are used in the network!
\end{itemize}

The remainder of this paper is organized as follows. In Section
\ref{sec:review}, we give an overview of capacity analysis for NC,
MPT, MPR, and other existing techniques. In Section \ref{sec:def},
we introduce the models we used. In Section \ref{sec:mpt-mpr} and
\ref{sec:nc-mpt-mpr}, we give our main results with MPT and MPR when
network coding is not used and used respectively. We conclude our paper in
Section \ref{sec:dis}.

\section{Literature Reviews} \label{sec:review}

Gupta and Kumar in their seminal paper \cite{GuKu00} proved that the
throughput capacity in wireless ad hoc network is not scalable.
Subsequently, many researchers have focused on identifying
techniques that could alter this conclusion.
Recently, Ozgur et al. \cite{OLT07} proposed a hierarchical
cooperation technique based on virtual MIMO to achieve linear per
source-destination capacity. Cooperation can be extended to the
simultaneous transmission and reception at the various nodes in the
network, which is called {\em many-to-many communication} and can
result in significant improvement in capacity \cite{MSG07}.

Since the original work by Ahlswede et al.  \cite{ANLY00}, most of the
research on network coding has focused on directed networks, where
each communication link has a fixed direction.  Li and Li
\cite{LiLi04} were the first to study the benefits of network coding
in undirected networks, where each communication link is
bidirectional. Their result \cite{LiLi04} shows that, for a single
unicast or broadcast session, there are no improvement with respect
to throughput due to network coding. In the case of a single
multicast session, such an improvement is bounded by a factor of
two. Meanwhile, \cite{RSW05, AKMK07, KASYK07} studied the throughput
capacity of NC in wireless ad hoc networks. However \cite{RSW05,
AKMK07, KASYK07} employ network models that are fundamentally
inconsistent with the more commonly accepted assumptions of ad-hoc
networks \cite{GuKu00}. Specifically, the model constraints of
\cite{LiLi04, LLL06, RSW05, AKMK07, KASYK07} differ as follows: All
the prior works assume a single source for unicast, multicast or
even broadcast. Aly et al. \cite{AKMK07} and Kong et al.
\cite{KASYK07} differentiate the total nodes into source set, relay
set and destination set. They do not allow all of the nodes to
concurrently serve as sources, relays or destinations, as allowed in
the work by Gupta and Kumar \cite{GuKu00}. An even bigger limitation
of these results is that they do not consider the impact of
interference in wireless ad hoc networks.

In the absence of interference, the communication
scenario equates an ideal case where a node can simultaneously  transmit
and receive from multiple nodes. Interference cancelation techniques
such as  MPT and MPR indeed enable nodes with the ability of
multi-point communication within a communication range of $T(n)$.
Thus, the model assumptions in \cite{RSW05, AKMK07, KASYK07} at the
very least assume that nodes are capable of MPT and MPR. Similarly,
works such as Physical-Layer Network Coding (PNC) \cite{ZLL06}
and Analog Network Coding \cite{KGK07}
also implicitly assume the ability of MPT and MPR.

\section{Network Model, Definitions, and Preliminaries}\label{sec:def}

We assume a random wireless ad hoc network with $n$ nodes
distributed uniformly in a unit-square network area. Our capacity
analysis is based on the protocol model for dense networks,
introduced by Gupta and Kumar \cite{GuKu00}. The case of what we
call point-to-point communication corresponds to the original
protocol model.

\begin{definition} The Protocol Model of Point-to-Point Communication:
All nodes use a common transmission range $r(n)$ for all their
communication.
Node $X_i$ can successfully transmit to node $X_j$ if for any node
$X_k, k \neq i$, that transmits at the same time as $X_i$ it is true
that $|X_i - X_j |  \le r(n)$ and $|X_k - X_j| \ge (1 +
\Delta)r(n)$.
\end{definition}

We make the following extensions to account for MPT and MPR
capabilities at the transmitters and receivers, respectively. In
wireless ad hoc networks with MPT (MPR) capability, any transmitter (receiver) node
can transmit (receive) different information simultaneously to (from) multiple nodes
within the circle whose radius is $T(n)$ \cite{GSW07b}. We further assume that nodes
cannot transmit and receive at the same time, which is equivalent to half-duplex communications
\cite{GuKu00}. From system point of view, MPT and MPR are dual if we
consider the source and destination duality.

\begin{definition} {\em Feasible throughput capacity \\}
In a wireless ad hoc network of $n$ nodes where each source
transmits its packets to $m$ destinations, a throughput of $C_m(n)$
bits per second for each node is feasible if there is a spatial and
temporal scheme for scheduling transmissions, such that, by operating
the network in a multi-hop fashion and buffering at intermediate
nodes when awaiting transmission, every node can send $C_m(n)$ bits
per second on average to its $m$ chosen destination nodes. That is,
there is a $T < \infty$ such that in every time interval $[(i-1)T,
iT ]$ every node can send $TC_m(n)$ bits to its corresponding
destination nodes.
\end{definition}

\begin{definition} {\em Order of throughput capacity }\\
$C_{m}(n)$ is said to be of order $\Theta(f(n))$ bits/second if
there exist deterministic positive constants $c$ and $c'$ such that
\begin{equation} \label{eq:throughput_capacit}
\left\{ \begin{aligned}
         \lim_{n\rightarrow\infty}\textrm{ Prob }(C_{m}(n)=cf(n) \textrm{ is feasible}) &= 1 \\
                    \lim_{n\rightarrow\infty}\textrm{ Prob }(C_{m}(n)=c'f(n) \textrm{ is
                    feasible})&<1.
                          \end{aligned} \right.
\end{equation}
\end{definition}

\begin{definition} {\em Euclidean Minimum Spanning Tree (EMST)} \\
Consider a connected undirected graph $G=(V, E)$, where $V$ and $E$
are sets of vertices and edges in the graph $G$, respectively.  The
EMST of $G$ is a spanning tree of $G$ with the minimum sum of
Euclidean distances between connected vertices of this tree.
\end{definition}

\begin{definition} {\em Minimum Euclidean Multicast Tree \\
$\left(\textrm{MEMT}\left(T(n)\right)\right)$:~} The
MEMT$\left(T(n)\right)$ is a multicast tree in which the $m$
destinations receive information from the source and this multicast
tree has the minimum total Euclidean distance.
\end{definition}

\begin{definition} {\em Minimum Area Multicast Tree~\\
$\left(\textrm{MAMT}\left(T(n)\right)\right)$:~} The
MAMT$\left(T(n)\right)$ in a multicast tree with $m$ destinations
for each source is a multicast tree that has minimum total area.
Area of a multicast tree is defined as the total area covered by
circles centered around each source or relay with radius of $T(n)$
(see Fig. \ref{fig:multicasttree}).
\end{definition}

\begin{figure}[http]
    \center
      \includegraphics[width=2.3in]{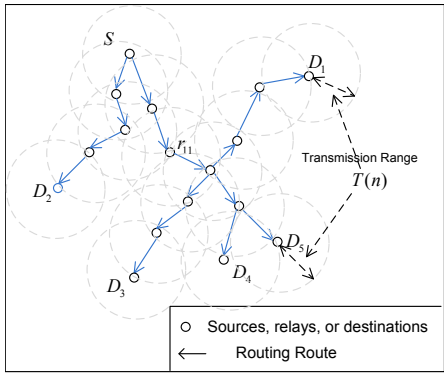}\\
      \caption{Area Coverage By One Multicast Tree}
    \label{fig:multicasttree}
\end{figure}

Note that EMST and MEMT are spanning trees which includes only
source and destinations, while MAMT is related to a real routing
tree which includes relays.

\begin{definition} {\em Total Active Area~
$\left(\textrm{TAA}\left(\Delta, T(n)\right)\right)$:~}\\
 The
TAA$\left(\Delta,T(n)\right)$ is the total  area of the network
multiplied by the average maximum number of simultaneous
transmissions and receptions inside a communication region of
$\Theta(T^2(n))$.
\end{definition}

It can be shown that this value has an upper bound of $O(1)$,
$O(nT^2(n))$ and $O(n^2 T^4(n))$ for point-to-point, MPR (or MPT)
and MPR combined with MPT respectively.

In the rest of this paper, $\| T \|$ denotes the total Euclidean
distance of a tree $T$; $\# T$ is used to denote the total number of
vertices (nodes) in a tree $T$; $S(T)$ denotes the area of tree $T$
covered; and $\overline{\| T \|}$ is used for the statistical
average of the total Euclidean distance of a tree.

To compute the multicast capacity, we use the relationship between
MAMT and EMST. Steele \cite{St88} determined a tight bound for
$\overline{\|\textrm{EMST}\|}$ in the following lemma.

\begin{lemma}\label{lem:EMST}
Let $f(x)$ denote the node probability distribution function in the
network area. Then, for large values of $m$ and $d > 1$, the
$\overline{\|\textrm{EMST}\|}$ is tight bounded as
\begin{eqnarray}\label{eq:EMST}
     \overline{\|\textrm{EMST}\|} &=&
     \Theta\left(c(d)m^{\frac{d-1}{d}}\int_{R^d}f(x)^{\frac{d-1}{d}}dx\right),
\end{eqnarray}
where $d$ is the dimension of the network. Note that both $c(d)$ and
the integral are constant values and not functions of $m$. When
$d=2$, then $\overline{\|\textrm{EMST}\|}= \Theta \left(
\sqrt{m}\right)$.
\end{lemma}

Given that the distribution of nodes in a random network is uniform,
if there are $n$ nodes in a unit square, then the density of nodes
equals $n$. Hence, if $|S|$ denotes the area of space region S, the
expected number of the nodes, $E(N_S)$, in this area is given by
$E(N_S) = n|S|$. Let $N_j$ be a random variable defining the number
of nodes in $S_j$ . Then, for the family of variables $N_j$, we have
the following standard results known as the Chernoff
bounds~\cite{MoRa95}:

\begin{lemma} Chernoff bound

For any $0 < \delta < 1$, we have
\begin{equation}\label{eqn:chernoff}
P[|N_j-n|S_j|| > \delta n|S_j|] < e^{-\theta n|S_j|}.
\end{equation}
\end{lemma}

Therefore, for any $\theta> 0$, there exist constants such that
deviations from the mean by more than these constants occur with
probability approaching zero as $n\rightarrow \infty$. It follows
that, w.h.p., we can get a very sharp concentration on the number of
nodes in an area, so we can find the achievable lower bound w.h.p.,
provided that the upper bound (mean) is given. In the following
sections, we first derive the upper bound, and then use the Chernoff
bound to prove the achievable lower bound.

In \cite{ANLY00} it was proved that the max-flow min-cut is equal to
multicast capacity of a directed graph with single source. The
directed graph model is more applicable for wired networks. However,
in this work we wish to study the utility of NC in a wireless
environment where links are bidirectional \cite{RSW05, AKMK07}.

In a single-source network, the cut capacity is equal to the maximum
flow. Thus \cite{AKMK07} provides an upper bound on the multicast
capacity of a network with single source and NC+MPT+MPR capability.
However, in \cite{RSW05, AKMK07, KASYK07}, the source, relays and
destinations are strictly different and information can not be
transmitted directly towards the destinations. These two assumptions
will be eventually relaxed in this paper.

\section{The Throughput Capacity with MPT and MPR}\label{sec:mpt-mpr}

We now analyze the scaling laws in random
geometric graphs with MPT and MPR abilities.  Wang et
al. \cite{WSG08a} proved the unifying capacity with point-to-point
communication, which resolves the general multicast case with $m$
destinations for each source being a function of $n$. Here, we use a
similar approach to prove the capacity with MPT and MPR when $m$ is
not a function of $n$ but a constant.
\subsection{Upper Bound} \label{subsec:upper}

The following Lemma provides an upper bound for the per-session
capacity as a function of $\overline{\textrm{TAA}(\Delta, T(n))}$
and $\overline{\textrm{MAMT}\left(T(n)\right)}$. Essentially,
$\overline{S\left(\textrm{MAMT}(T(n))\right)}$ equals the minimum
area consumed to multicast a packet to $m$ destinations (see Fig.
\ref{fig:multicasttree}), and $\overline{\textrm{TAA}(\Delta,
T(n))}$ represents the maximum area which can be supported when MPT
and MPR are used.

\begin{lemma}\label{lem:important1}
In random dense wireless ad hoc networks, the per-node throughput
capacity of multicast with MPT and MPR is given by
$O\left(\frac{1}{n} \times \frac{\overline{\textrm{TAA}(\Delta,
T(n))}}{\overline{S(\textrm{MAMT}(T(n))})}\right)$.
\end{lemma}

\begin{IEEEproof}
With MPT and MPR, we observe that $\overline{S\left(
\textrm{MAMT}(T(n))\right)}$ represents the total area required to
transmit information from a multicast source to all its $m$
destinations. The ratio between average total active area,
$\overline{\textrm{TAA}(\Delta, T(n))}$, and
$\overline{S\left(\textrm{MAMT}(T(n))\right)}$ represents the
average number of simultaneous multicast communications that can
occur in the network. Normalizing this ratio by $n$ provides
per-node capacity.
\end{IEEEproof}

Lemma \ref{lem:important1} provides the upper bound for the
multicast throughput capacity with MPT and MPR as a function of
$\overline{S\left(\textrm{MAMT}(T(n))\right)}$ and
$\overline{\textrm{TAA}(\Delta, T(n))}$. In order to compute the
upper bound, we derive the upper bound of
$\overline{\textrm{TAA}(\Delta, T(n))}$ and the lower bound of
$\overline{S\left(\textrm{MAMT}(T(n))\right)}$. Combining these
results provides an upper bound for the multicast throughput
capacity with MPT and MPR.

\begin{lemma}\label{lem:prepared}
 The average area of a multicast tree with transmission range $T(n)$,
$\overline{S\left(\textrm{MAMT}(T(n))\right)}$ is lower bounded by
$\Omega \left(T(n)\right)$, when $m$ is a constant value.
\end{lemma}

\begin{IEEEproof}
From \cite{LiTaFr07}, it can be deduced that
$\overline{S\left(\textrm{MAMT}(T(n))\right)}$ is lower bounded as
$\Omega \left( \overline{\|\textrm{EMST}\|} \times T(n)\right)$.
Even for the case of the minimum value for $T(n)$ to assure
connectivity, this upper bound is guaranteed for constant values of
$m$. Lemma \ref{lem:EMST} states that $\overline{\|\textrm{EMST}\|}
= \Theta(\sqrt{m}) = \Theta(1)$. The proof follows immediately.
\end{IEEEproof}

\begin{lemma}\label{lem:jj-added}
The average total active area, $\overline{\textrm{TAA}(\Delta,
T(n))}$, has the following upper bound in networks with MPT and MPR.
\begin{equation}\label{eq:MMMIS_upperbound}
\overline{\textrm{TAA}(\Delta, T(n))} = O\left( n^2T^4(n)\right)
\end{equation}
\end{lemma}

\begin{IEEEproof}
As discussed earlier, the $\overline{\textrm{TAA}(\Delta, T(n))}$
for point-to-point communication is equal to 1 since for each circle
of radius $T(n)$, there is only a single pair of
transmitter-receiver nodes (see Fig. \ref{fig:area12}). For the case
of MPR and MPT, the number of nodes in a circle of radius $T(n)$ is
upper bounded as $O(n T^2(n))$. This is also upper bound for the
number of transmitters or receivers in this region. The upper bound
for  $\overline{\textrm{TAA}(\Delta, T(n))}$ is achieved when the
maximum number of transmitter and receivers are employed in this
circle. Figure \ref{fig:area12} demonstrates an example that can
achieve this upper bound simultaneously for transmitters and
receivers. Let a circle of radius $\frac{T(n)}{2}$ located at the
center of another circle of radius $T(n)$. Note that with this
construction,  any two nodes inside the small circle are connected.
If we randomly assign half of the nodes inside the smaller circle as
transmitters and the other half as receiver nodes, then the average
number of transmitters and receivers in this circle are proportional
to $\Theta(n T^2(n))$. Given the fact that this value also is the
maximum possible number of transmitter and receiver nodes, the
result follows immediately.
\end{IEEEproof}

\begin{figure}[http]
    \center
      \includegraphics[width=2.3in]{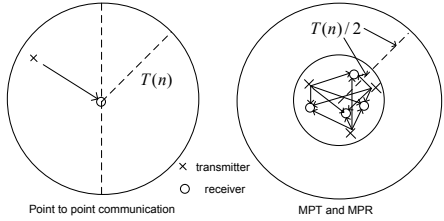}\\
      \caption{Upper Bound of Total Available Area Based On Protocol Model}
    \label{fig:area12}
\end{figure}

Combining Lemmas \ref{lem:important1}, \ref{lem:prepared} and
\ref{lem:jj-added}, we can compute the upper bound for multicast
capacity of MPT and MPR.

\begin{theorem}\label{the:upper_bound}
In wireless ad hoc networks with MPT and MPR, the upper bound on the
per-node throughput capacity of multicast with constant number of
destinations is
\begin{equation} \label{eq:upper_bound}
C_{m}(n) = O\left( nT^3(n)\right)
\end{equation}
\end{theorem}

\subsection{Lower Bound} \label{subsec:lower}

To derive an achievable lower bound, we use a TDMA scheme for random
dense wireless ad hoc networks similar to the approach used in
\cite{GiKu05,KuVi04}.

We first divide the network area into square cells.  Each square
cell has an area of $T^2(n)/2$, which makes the diagonal length of
square equal to $T(n)$, as shown in Fig.~\ref{fig:construction}.
Under this condition, connectivity inside all cells is guaranteed
and all nodes inside a cell are within communication range of each
other. We build a cell graph over the cells that are occupied with
at least one vertex (node). Two cells are connected if there exist a
pair of nodes, one in each cell, that are less than or equal to
$T(n)$ distance apart. Because the whole network is connected when
$T(n) = r(n) \geq \Theta\left(\sqrt{ \log n /n} \right)$, it follows
that the cell graph is connected \cite{GiKu05, KuVi04}.
\begin{figure}[http]
    \center
      \includegraphics[width=2.3in]{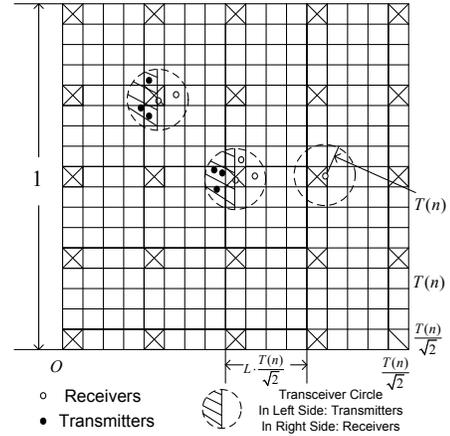}\\
      \caption{Cell construction used to derive a lower bound on capacity}
    \label{fig:construction}
    \end{figure}

To satisfy the MPT and MPR protocol model, we organize cells in
groups so that simultaneous transmissions within each group does not
violate the conditions for successful communication in the MPT and
MPR protocol model.  Let $L$ represent the minimum number of cell
separations in each group of cells that communicate simultaneously.
Utilizing the protocol model, $L$ satisfies the following condition:
\begin{equation}\label{eq:far_away}
    L= \left\lceil 1+\frac{T(n)+(1+\Delta) T(n)}{T(n)/\sqrt{2}}\right\rceil=
\lceil 1+\sqrt{2}(2+\Delta) \rceil.
\end{equation}

If we divide time into $L^2$ time slots and assign each time slot to
a single group of cells, interference is avoided and the protocol
model is satisfied.  The separation example can be shown for the
upper two receiver circles in Fig. \ref{fig:construction}.  For the
MPT and MPR protocol model, the distance between two adjacent
receiving nodes is $(2+\Delta)T(n)$. Because this distance is
smaller than $(L-1)T(n)$, this organization of cells  guarantees
that the MPT and MPR protocol model is satisfied.  Fig.
\ref{fig:construction} represents one of these groups with a cross
sign inside those cells for $L=4$. We can derive an achievable multicast capacity for MPT and MPR by
taking advantage of this cell arrangement and TDMA scheme.
The capacity reduction caused by the TDMA scheme is a constant
factor and does not change the order capacity of the network.

Next our objective is to find an achievable lower bound using the
Chernoff bound, such that the distribution of the number of edges in
this unit space is sharply concentrated around its mean, and hence
the actual number of simultaneous transmissions occurring in the
unit space in a randomly chosen network is indeed
$\Theta(n^2T^2(n))$ w.h.p..

\begin{lemma}\label{lem:lower bound}
{The circular area of radius $T(n)$ corresponding to the transceiver
range of any node $j$ in the cross area in Fig.
\ref{fig:construction} contains $\Theta(n T^2(n))$ nodes w.h.p., and
is uniformly distributed for all values of $j$, $1 \leq j \leq
\frac{1}{(LT(n)/\sqrt{2})^2}$}.
\end{lemma}

\begin{IEEEproof}
The statement of this lemma can be expressed as {\small
\begin{equation}\label{31}
    \lim_{n\rightarrow \infty}P\left[\bigcap_{j=1}^{\frac{1}{(LT(n)/\sqrt{2})^2}}|N_j-E(N_j)|<\delta E(N_j)\right]=1,
\end{equation}
}

\noindent where $N_j$ and $E\left(N_j\right)$ are the  random
variables that represent the number of transmitters in the receiver
circle of radius $T(n)$ centered by the receiver $j$ and the
expected value of this random variable respectively, and $\delta$ is
a positive arbitrarily small value close to zero.

From the Chernoff bound in Eq. (\ref{eqn:chernoff}), for any given
$0<\delta<1$, we can find $\theta>0$ such that
$ P\left[|N_j-E(N_j)|>\delta E(N_j)\right] < e^{-\theta E(N_j)}$.
Thus, we can conclude that the probability that the value of the
random variable $N_j$ deviates by an arbitrarily small constant
value from the mean tends to zero as $n\rightarrow \infty$. This is
a key step in showing that when all the events
$\bigcap_{j=1}^{\frac{1}{\left(LT(n)/\sqrt{2}\right)^2}}|N_j-E(N_j)|<\delta
E(N_j)$ occur simultaneously, then all $N_j$'s converge uniformly to
their expected values. Utilizing the union bound, we arrive at
\begin{eqnarray}\label{33}
  &&P\left[\bigcap_{j=1}^{\frac{1}{(L T(n)/\sqrt{2})^2}}|N_j-E(N_j)|<\delta E(N_j)\right]\nonumber \\
   &\geq& 1-\sum_{j=1}^{\frac{1}{(L T(n)/\sqrt{2})^2}}P\left[|N_j-E(N_j)|>\delta
   E(N_j)\right]\nonumber
   \\
&>&1-\frac{1}{(L T(n)/\sqrt{2})^2}e^{-\theta E(N_j)}.
\end{eqnarray}

Given that $E(N_j)=\pi n T^2(n)$, then we have
\begin{eqnarray}\label{34}
    &&\lim_{n\rightarrow \infty}P\left[\bigcap_{j=1}^{\frac{1}{(L T(n)/\sqrt{2})^2}}|N_j-E(N_j)|<\delta
    E(N_j)\right]\nonumber\\
    &&\geq 1-\lim_{n\rightarrow \infty} \frac{1}{(L T(n)/\sqrt{2})^2}e^{-\theta\pi n T^2(n)}
\end{eqnarray}

Utilizing the connectivity criterion,
$\lim_{n\rightarrow \infty}\frac{e^{-\theta\pi
nT^2(n)}}{T^2(n)}\rightarrow 0$, which finishes the proof.
\end{IEEEproof}

Furthermore, we can arrange all of the nodes in the left side of the
corresponding transceiver circle be the transmitters, and all of the
nodes in the right side of the corresponding transceiver circle be
the receivers. Thus, we arrive at the following lemma.

    \begin{lemma}\label{lem:total}
    In the unit square area for a wireless ad hoc network shown in Fig.
    \ref{fig:construction}, the total number of transmitter-receiver
    links (simultaneous transmissions) is
    $\Omega\left(n^2T^2(n)\right)$.
    \end{lemma}

    \begin{IEEEproof}
    From Lemma \ref{lem:lower bound}, for any node in the cross cell
    in the whole network shown in Fig. \ref{fig:construction}, there are
    $\Theta(nT^2(n))$ nodes in the transceiver circle. We divided the
    total nodes into two categories, transmitters in the left of the
    transceiver circles and receivers in the right of the transceiver
    circles. To guarantee all of the transmitters and receivers are in
    the transceiver range, we only consider the nodes in the circle with
    radius $T(n)/2$. Because of the MPT and MPR capabilities, so that
    every transmitter in the left of the transceiver circle with
    $T(n)/2$ radius can transmit successfully to every receiver in the
    right, then the total number of successful transmissions is
    $\pi^2n^2T^4(n)/16$ which is the achievable lower bound. The actual
    number of the transmissions can be much larger than this because we
    only consider $T(n)/2$ instead of $T(n)$. Using the Chernoff Bound
    in Eq. \ref{eqn:chernoff} and Lemma \ref{lem:lower bound}, we can
    get w.h.p. that the total number of successful transmissions is
    \begin{equation}
        \Omega\left(\frac{1}{\left(LT(n)/\sqrt{2}\right)^2}\times
        \frac{\pi^2n^2T^4(n)}{16}\right)=
        \Omega\left(n^2T^2(n)\right)
    \end{equation}
    \end{IEEEproof}

    The above results enables us to obtain the following achievable lower bound.

    Let us define $\overline{\#\textrm{MEMTC}(T(n))}$ as the total
    number of cells that contain all the nodes in a multicast group. The
    following lemma establishes the achievable lower bound for the
    multicast throughput capacity of MPT and MPR as a function of
    $\overline{\#\textrm{MEMTC}(T(n))}$.

    \begin{lemma}\label{lem:memtc}
    The achievable lower bound of the multicast capacity is given by
    \begin{equation}\label{eq:lower_bound_important2}
        C_{m}(n)=\Omega \left(\frac{nT^2(n)}{\overline{\#\textrm{MEMTC}(T(n))}}
        \right).
    \end{equation}
    \end{lemma}
    \begin{IEEEproof}
    There are $(T(n)/\sqrt{2})^{-2}$ cells in the unit square network
    area. From the definition of $\overline{\#\textrm{MEMTC}(T(n))}$ and
    the fact that our TDMA scheme does not change the order capacity, it is clear that there are at most in the
    order of $\overline{\#\textrm{MEMTC}(T(n))}$ interfering cells for
    multicast communication. Hence, from Lemma \ref{lem:total}, there
    are a total of $\Theta\left(n^2T^2(n)\right)$ nodes transmitting
    simultaneously, which are distributed over all the
    $(T(n)/\sqrt{2})^{-2}$ cells. For each cell, the order of nodes in
    each cell is $\Omega\left(n^2T^4(n)\right)$. Accordingly, the total
    lower bound capacity is given by $\Omega\left(
    \left(T(n)/\sqrt{2}\right)^{-2}\times\left(n^2
    T^4(n)\right)\times\left(\overline{\#\textrm{MEMTC}(T(n))}\right)^{-1}\right)$.
    Normalizing this value by total number of nodes in the network, $n$,
    proves the lemma.
    \end{IEEEproof}

    Given the above lemma, to express the lower bound of $ C_{m}(n)$ as
    a function of network parameters, we need to compute the upper bound
    of $\overline{\#\textrm{MEMTC}(T(n))}$, which we do next.

    \begin{lemma}\label{lem:rela}
    The average number of cells covered by the nodes in
    $\textrm{MEMTC}(T(n))$, is upper bounded w.h.p. as follows:
    \begin{equation}\label{eq:lower_bound_rela}
    \overline{\#\textrm{MEMTC}(T(n))}\le
    \Theta\left(\frac{\sqrt{m}}{T(n)}\right)
    \end{equation}
    \end{lemma}

    \begin{IEEEproof}
    Because $T(n)$ is the transceiver range of the network, the maximum
    number of cells for this multicast tree must be at most
    $\Theta\left(\sqrt{m}T^{-1}(n)\right)$, i.e.,
    $\overline{\#\textrm{MEMTC}\left(T(n)\right)}\le
    \Theta\left(\sqrt{m}T^{-1}(n)\right)$.  This upper bound can be achieved
    only if every two adjacent nodes in the multicast tree belong to two
    different cells in the network. However, in practice, it is possible
    that some adjacent nodes in multicast tree locate in a single cell.
    Consequently, this value is upper bound as described in
    (\ref{eq:lower_bound_rela}).
    \end{IEEEproof}

    Combining Lemmas \ref{lem:memtc} and \ref{lem:rela}, we arrive at
    the achievable lower bound of the multicast throughput capacity in
    dense random wireless ad hoc networks with MPT and MPR.
    \begin{theorem}\label{the:lower_bound}
    When the number of the destinations $m$ is a constant, the
    achievable lower bound of the $m$ multicast throughput capacity with
    MPT and MPR is
    \begin{equation}\label{eq:nmk_lower_lower} C_{m}(n) =
    \Omega\left(\frac{nT^3(n)}{\sqrt{m}}\right)
    \end{equation}
    \end{theorem}

    \subsection{Tight Bound and Comparison with Point-to-Point Communication}

    From Theorems \ref{the:upper_bound} and \ref{the:lower_bound},  we
    can provide a tight bound throughput capacity for  multicasting
    when nodes have MPT and MPR capabilities in dense random wireless
    ad hoc networks as follows.

    \begin{theorem}\label{the:tight_bound}
    The throughput capacity of multicast in random dense wireless ad
    hoc network with MPT and MPR is
    \begin{equation}\label{eq:tight_bound}
    C_{m}^{\textrm{MPT+MPR}}(n)= \Theta\left(\frac{nT^3(n)}{\sqrt{m}}\right)
    \end{equation}
    The transceiver range of MPT and MPR should satisfy
    $T(n)\ge\Theta\left(\sqrt{\log n/n}\right)$.
    \end{theorem}

    The multicast throughput capacity with point-to-point communication is given by
   the following lemma \cite{WSG08a}.

    \begin{lemma}\label{lem:routing-capacity}
        In multicast with a constant number $m$ of destinations, without MPR or
        MPR ability, the capacity is
    \begin{equation}\label{eq:capacity-plainrouting}
         C_m^{\textrm{Routing}}(n)=\Theta\left(\frac{1}{\sqrt{m}nr(n)}\right)
    \end{equation}
    where, $r(n)\ge \Theta\left(\sqrt{\log n/n}\right)$. When
    $r(n)=\Theta\left(\sqrt{\log n/n}\right)$ for the minimum
    transmission range to guarantee the connectivity, then we obtain the
    maximum capacity as
    $C_m^{\textrm{Routing-Max}}(n)=\Theta\left(\frac{1}{\sqrt{m n\log
    n}}\right)$.
    \end{lemma}

    Combining Theorem \ref{the:tight_bound} with Lemma
    \ref{lem:routing-capacity}, the gain of throughput capacity with MPT
    and MPR capability in wireless ad hoc networks can be stated as follows.

    \begin{theorem}
        In multicast with a constant number $m$ of destinations, with
        MPT and MPR ability, the gain of per-node throughput capacity compared with point-to-point communication is
    $\Theta\left(n^2T^4(n)\right)$, where, $T(n)=
    r(n)\ge\Theta\left(\sqrt{\log n/
        n}\right)$. When $T(n)=\Theta\left(\sqrt{\log n/
        n}\right)$, the gain of per-node capacity is at least
    $\Theta\left(\log^2
            n\right)$.
    \end{theorem}

\section{Capacity with NC, MPT and MPR}\label{sec:nc-mpt-mpr}

We now study the multi-source multicast capacity of a wireless network
{\em in the absence of interference} when nodes use NC. The results we
present  serve as an upper-bound for what can be achieved
by combining NC, MPT and MPR in the presence of interference. Our
arguments are generic and can be used to deduce upper bounds for the
multicast capacity of other interesting cases where NC is used along
with only one of MPT or MPR, or even the scenario where NC is used
with traditional single packet transmission and reception.

We deduce the bounds for the case of multi-source multicasting by reducing it to a
suitable unicast routing problem. Under the reduction, an upper
bound for the unicast problem also serves for the original multicast
routing problem. Thus consider the following simple yet powerful
lemma

\begin{lemma} \label{multi-to-uni}
Consider a network with $n$ nodes $V = \{a_1, \dots, a_n\}$ and $k$
multicast sessions. Each session consists of one of the $n$ nodes
acting as a source with an arbitrary finite subset of $V$ acting as the set
of destinations. Let $s_i$ be the source of the $i^{th}$ session and
let $D_i = \{d_{i1}, \dots, d_{im_i}\}$ be the set of $m_i$
destinations. Now, there exists a joint routing-coding-scheduling
scheme that can realize a throughput of $\lambda_i$ for the $i^{th}$
session, i.e. $\lambda = [\lambda_1, \dots \lambda_k]$ is a
feasible rate vector. Then $\lambda$ is also a feasible vector for
any unicast routing problem in the same network such that the
traffic consists of $k$ unicast sessions with $s_i$ being the source
of the $i^{th}$ session and the destination $b_i$ is any arbitrary element of the set
$D_i$.
\end{lemma}

If a multicast capacity from a source to multiple destinations is
feasible, then clearly it is feasible to achieve the same capacity
to only any single node from this set of destinations.

\begin{lemma} \label{n-eachside}
Consider a random geometric network with $n$ nodes distributed
uniformly in a unit square. Consider a
decomposition of the unit-square into two disjoint regions $R$ and
$R^c$ such that the area of each region is of order $\Theta(1)$. Now
consider a multicast traffic scenario consisting of $n$ sessions
with each node being the source of a session  and $m$ randomly
chosen nodes being the destination of the session. We say that a
source satisfies property $P$ if the source belongs to region $R$
and at least one of its destination belongs to $R^c$ OR if the source
belongs to region $R^c$ and at least one of its destination belongs
to $R$. It can be shown that w.h.p the number of sources satisfying
property $P$ are $\Theta(n)$
\end{lemma}

\begin{theorem} \label{the:nc-mpt-mpr}
In a wireless ad hoc network formed by $n$ nodes distributed randomly
in a unit square with traffic formed by each node acting as source for
a multicast sessions with $m=\Theta(1)$ randomly chosen nodes as
destinations, the per-session multicast capacities are
    \begin{eqnarray}
    C_m^{\textrm {NC+PTP}} &=& \Theta\left(\frac{1}{nT(n)}\right) \\
    C_m^{\textrm{NC+MPT}} = C_m^{\textrm{NC+MPR}} &=& \Theta\left(T(n)\right)\\
    C_m^{\textrm{NC+MPT+MPR}} &=& \Theta\left(nT^3(n)\right)
    \end{eqnarray}
where NC + PTP denotes the use of NC with point-to-point communication (no MPT or MPR),
i.e., a node can only transmit or receive at most one packet at a time.
\end{theorem}
\begin{proof} (Sketch)
For any sparsity cut of the unit area as illustrated in Fig.
\ref{fig:proof}, the middle line induces a sparsity cut, lemmas
\ref{n-eachside}. \ref{multi-to-uni} tell us that we can construct a
unicast routing problem satisfying the property that any rate for
the unicast problem is feasible for the original multicast problem
and we have $\Theta(n)$ source-destination pairs across the cut.
Thus, the capacity of the sparsity cut provides an upper bound for
the unicast problem, which can in turn be used to provide a bound
for the multicast problem. Liu et al. \cite{LGT07} showed that the
maximum number of packets that can be simultaneously transmitted
across the cut is $\Theta\left(\frac{1}{T(n)}\right)$ for the case
of unicast with point-to-point communication and NC.  With similar
arguments, we can show that the combination of NC+MPT or NC+MPR
allows us to transmit a maximum of $\Theta\left(nT(n)\right)$
packets across the cut. Finally, we can extend such arguments to
show that the combination of NC+MPT+MPR allows us to simultaneously
transmit a maximum of $\Theta\left(n^2T^3(n)\right)$ packets across
the cut. The result of the theorem then follows from the fact that
the cut capacity has to be divided among the $\Theta(n)$
source-destination pairs across the cut.
\end{proof}
\begin{figure}
    \center
      \includegraphics[width=2.3in]{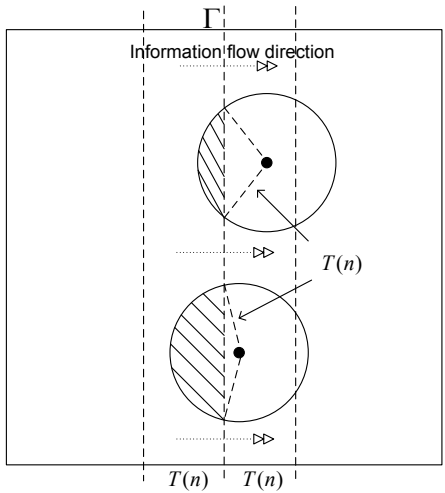}\\
      \caption{All the nodes in the shaded region can send a message simultaneously}
    \label{fig:proof}
    \end{figure}

\section{Discussion}\label{sec:dis}

By combining the
results from theorems \ref{the:tight_bound} and
\ref{the:nc-mpt-mpr}, the main contribution of this paper is stated in the following theorem.

\begin{theorem}\label{the:final}
In wireless ad hoc networks with multi-pair multicast sessions and
with a finite number of destinations for each source ($m$), the
throughput capacity utilizing NC, MPT and MPR capabilities for all
nodes is the same order as when the nodes are endowed only with MPT
and MPR.
\begin{equation}\label{eq:final}
    {C}_m^{\textrm{MPT+MPR+NC}}(n) ={C}_m^{\textrm{MPT+MPR}}(n)
\end{equation}
\end{theorem}

It is also important to emphasize that, as
Theorem~\ref{the:nc-mpt-mpr} shows, NC does not provide any order
capacity gain for multi-source multicasting when the size of receiver
groups is $m=\Theta(1)$ and nodes use point-to-point communication.
Hence,
the result in Theorem~\ref{the:final}
implies that NC does not provide an order capacity gain
when either MPT and MPR are used, or point-to-point communication is
used, and that MPT and MPR are the real contributing factor for order
capacity increases in wireless networks.

\section{Acknowledgments}
This work was partially sponsored by the U.S. Army Research Office
under grants W911NF-04-1-0224 and W911NF-05-1-0246, by the National
Science Foundation under grant CCF-0729230, by the Defense Advanced
Research Projects Agency through Air Force Research Laboratory
Contract FA8750-07-C-0169, and by the Baskin Chair of Computer
Engineering. The views and conclusions contained in this document
are those of the authors and should not be interpreted as
representing the official policies, either expressed or implied, of
the U.S. Government.

\bibliographystyle{IEEEtran}
\bibliography{database}

\begin{thebibliography}{10}
\providecommand{\url}[1]{#1}
\csname url@samestyle\endcsname
\providecommand{\newblock}{\relax}
\providecommand{\bibinfo}[2]{#2}
\providecommand{\BIBentrySTDinterwordspacing}{\spaceskip=0pt\relax}
\providecommand{\BIBentryALTinterwordstretchfactor}{4}
\providecommand{\BIBentryALTinterwordspacing}{\spaceskip=\fontdimen2\font plus
\BIBentryALTinterwordstretchfactor\fontdimen3\font minus
  \fontdimen4\font\relax}
\providecommand{\BIBforeignlanguage}[2]{{%
\expandafter\ifx\csname l@#1\endcsname\relax
\typeout{** WARNING: IEEEtran.bst: No hyphenation pattern has been}%
\typeout{** loaded for the language `#1'. Using the pattern for}%
\typeout{** the default language instead.}%
\else
\language=\csname l@#1\endcsname
\fi
#2}}
\providecommand{\BIBdecl}{\relax}
\BIBdecl

\bibitem{GuKu00}
P.~Gupta and P.~R. Kumar, ``The capacity of wireless networks,'' \emph{IEEE
  Transactions on Information Theory}, vol.~46, no.~2, pp. 388--404, 2000.

\bibitem{OLT07}
A.~Ozgur, O.~Leveque, and D.~Tse, ``Hierarchical cooperation achieves optimal
  capacity scaling in ad hoc networks,'' \emph{IEEE Transactions on Information
  Theory}, vol.~53, no.~10, pp. 2549--3572, 2007.

\bibitem{MSG07}
R.~M. de~Moraes, H.~R. Sadjadpour, and J.~J. Garcia-Luna-Aceves, ``Many-to-many
  communication: A new approach for collaboration in manets,'' in \emph{Proc.
  of IEEE INFOCOM 2007}, Anchorage, Alaska, USA., May 6-12 2007.

\bibitem{GMPS07}
A.~El~Gamal, J.~Mammen, B.~Prabhakar, and D.~Shah, ``Optimal throughput-delay
  scaling in wireless networks part i: The fluid model,'' \emph{IEEE
  Transactions on Information Theory}, vol.~52, no.~6, pp. 2568--2592, 2007.

\bibitem{ANLY00}
R.~Ahlswede, C.~Ning, S.-Y.~R. Li, and R.~W. Yeung, ``Network information
  flow,'' \emph{IEEE Transactions on Information Theory}, vol.~46, no.~4, pp.
  1204--1216, 2000.

\bibitem{LiYeCa03}
S.-Y.~R. Li, R.~W. Yeung, and N.~Cai, ``Linear network coding,'' \emph{IEEE
  Transactions on Information Theory}, vol.~49, no.~2, pp. 371--381, 2003.

\bibitem{KoMe03}
R.~Koetter and M.~Medard, ``An algebraic approach to network coding,''
  \emph{IEEE/ACM Transactions on Networking}, vol.~11, no.~5, pp. 782--795,
  2003.

\bibitem{LGT07}
J.~Liu, D.~Goeckel, and D.~Towsley, ``Bounds on the gain of network coding and
  broadcasting in wireless networks,'' in \emph{Proc. of IEEE INFOCOM 2007},
  Anchorage, Alaska, USA., May 6-12 2007.

\bibitem{KGK07}
S.~Katti, S.~Gollakota, and D.~Katabi, ``Embracing wireless interference:
  Analog network coding,'' in \emph{Proc. of ACM SIGCOMM 2007}, Kyoto, Japan,
  August 27-31 2007.

\bibitem{ZLL06}
S.~Zhang, S.~Liew, and P.~P. Lam, ``Hot topic: Physical-layer network coding,''
  in \emph{Proc. of ACM MobiCom 2006}, Los Angeles, California, USA., September
  23-29 2006.

\bibitem{RSW05}
A.~Ramamoorthy, J.~Shi, and R.~Wesel, ``On the capacity of network coding for
  random networks,'' \emph{IEEE Transactions on Information Theory}, vol.~51,
  no.~8, pp. 2878--2885, 2005.

\bibitem{AKMK07}
S.~A. Aly, V.~Kapoor, J.~Meng, and A.~Klappenecker, ``Bounds on the network
  coding capacity for wireless random networks,'' in \emph{Proc. of Third
  Workshop on Network Coding, Theory, and Applications (NetCod 2007}, San
  Diego, California, USA, January 29 2007.

\bibitem{KASYK07}
\BIBentryALTinterwordspacing
Z.~Kong, S.~A. Aly, E.~Soljanin, E.~M. Yeh, and A.~Klappenecker, ``Network
  coding capacity of random wireless networks under a
  signal-to-interference-and noise model,'' \emph{Submitted to IEEE
  Transactions on Information Theory}, 2007. [Online]. Available:
  \url{http://arxiv.org/abs/0708.3070}
\BIBentrySTDinterwordspacing

\bibitem{ToGo03}
S.~Toumpis and A.~J. Goldsmith, ``Capacity regions for wireless ad hoc
  networks,'' \emph{IEEE Transactions on Wireless Communications}, vol.~2,
  no.~4, pp. 736--748, 2003.

\bibitem{GSW07b}
J.~J. Garcia-Luna-Aceves, H.~R. Sadjadpour, and Z.~Wang, ``Challenges: Towards
  truly scalable ad hoc networks,'' in \emph{Proc. of ACM MobiCom 2007},
  Montreal, Quebec, Canada, September 9-14 2007.

\bibitem{LiLi04}
Z.~Li and B.~Li, ``Network coding in undirected networks,'' in \emph{Proc. of
  CISS 2004}, Princeton, NJ, USA., March 17-19 2004.

\bibitem{LLL06}
Z.~Li, B.~Li, and L.~Lau, ``On achieving maximum multicast throughput in
  undirected networks,'' \emph{IEEE/ACM Transactions on Special issue on
  Networking and Information theory}, vol.~52, pp. 2467--2485, 2006.

\bibitem{St88}
M.~Steele, ``Growth rates of euclidean minimal spanning trees with power
  weighted edges,'' \emph{The Annals of Probability}, vol.~16, no.~4, pp.
  1767--1787, 1988.

\bibitem{MoRa95}
R.~Motwani and P.~Raghavan, \emph{Randomized Algorithms}.\hskip 1em plus 0.5em
  minus 0.4em\relax Cambridge University Press, 1995.

\bibitem{WSG08a}
Z.~Wang, H.~R. Sadjadpour, and J.~J. Garcia-Luna-Aceves, ``A unifying
  perspective on the capacity of wireless ad hoc networks,'' in \emph{IEEE
  INFOCOM 2008}, Phoenix, Arizona, USA., April 13-18 2008.

\bibitem{LiTaFr07}
X.-Y. Li, S.-J. Tang, and O.~Frieder, ``Multicast capacity for large scale
  wireless ad hoc networks,'' in \emph{Proc. of ACM MobiCom 2007}, Montreal,
  Canada, September 9-14 2007.

\bibitem{GiKu05}
A.~Giridhar and P.~R. Kumar, ``Computing and communicating functions over
  sensor networks,'' \emph{IEEE Journal on Selected Areas in Communications},
  vol.~23, no.~4, pp. 755--764, 2005.

\bibitem{KuVi04}
S.~Kulkarni and P.~Viswanath, ``A deterministic approach to throughput scaling
  wireless networks,'' \emph{IEEE Transactions on Information Theory}, vol.~50,
  no.~6, pp. 1041--1049, 2004.

\end{thebibliography}
\balance

\end{document}